\documentclass{article}
\usepackage{amssymb}
\usepackage{amsmath}
\usepackage{amsthm}
\usepackage{color}
\usepackage[colorlinks]{hyperref}
\usepackage[a4paper]{geometry}

\newtheorem{theorem}{Theorem}[section]
\newtheorem{lemma}[theorem]{Lemma}
\newtheorem{proposition}[theorem]{Proposition}
\newtheorem{corollary}[theorem]{Corollary}

\theoremstyle{definition}
\newtheorem{definition}{Definition}[section]
\newtheorem{example}{Example}[section]

\theoremstyle{remark}
\newtheorem{remark}{Remark}

\title{Frames over finite fields with respect to the Galois inner product and their connections to self-orthogonal and LCD codes}
\author{Junmin An\thanks{junmin0518@sogang.ac.kr, Department of Mathematics and Institute for Mathematical and Data Sciences, Sogang University, Seoul, Korea}, Jon-Lark Kim\thanks{jlkim@sogang.ac.kr, Department of Mathematics and Institute for Mathematical and Data Sciences, Sogang University, Seoul, Korea}}
\date{}

\begin{document}
\maketitle

\begin{abstract}
In this paper, we generalize the theory of frames to vector spaces over finite fields equipped with the $\ell$-Galois inner product. We introduce $\ell$-Galois analogues of fundamental concepts in frame theory and investigate extensions of existing results to the $\ell$-Galois setting. In particular, we establish necessary and sufficient conditions for a given matrix to be the $\ell$-Galois Gram matrix of a frame and show that such frames are unique up to isometry. We also extend Gerzon's bound to $\ell$-Galois equiangular systems and show that equality in this bound yields an $\ell$-Galois equiangular tight frame. Furthermore, we investigate the connections between $\ell$-Galois frames and $\ell$-Galois self-orthogonal and LCD codes over finite fields. We establish necessary and sufficient conditions for constructing $\ell$-Galois equiangular tight frames from self-dual codes, generalizing existing results by Shi et al. on the relationship between frames and self-dual codes over finite fields. Beyond self-dual codes, we provide sufficient conditions for constructing $\ell$-Galois equiangular tight frames from linear codes using their automorphism groups, together with explicit examples.\\

\end{abstract}

Frame theory, originally developed in the context of signal processing by Duffin and Schaeffer~\cite{DS-1952}, has become one of the central topics in mathematics, computer science, and engineering. Unlike a basis for a vector space which requires linear independence, a frame for a space provides a redundant, yet spanning set that allows for flexible representations of vectors. It is particularly advantageous in settings where robustness to noise, erasures, or perturbations are required, which makes frames especially useful in applications such as signal processing, image reconstruction, and data compression. In particular, Daubechies et al.~\cite{DGM-1986} established connections between frame theory and Gabor and wavelet analysis, which contributed to the development of frames as a tool for signal analysis and reconstruction. 

Beyond their applications, finite frames have also been studied from geometric and combinatorial perspectives~\cite{BF-2003,CMS-2017}. Among finite frames, equiangular tight frames (ETFs) are of particular interest because of their geometric optimality. Such frames yield optimal packings of lines in real or complex projective space~\cite{SH-2003}. These properties have motivated extensive research into the existence and construction of ETFs, which revealed connections between frame theory, graph theory, and combinatorial designs~\cite{BP-2005,FMT-2012,HP-2004}.

A recent development in frame theory is its extension to finite fields by Greaves et al.~\cite{GIJM-2022-1,GIJM-2022-2}. They extended the classical notion of equiangular tight frames to discrete settings. In this context, vector spaces over finite fields and corresponding equiangular systems are constructed and studied. Moreover, the resulting structures revealed connections to strongly regular graphs, finite classical groups, and incidence geometries. One application of frame theory over finite fields arises in coding theory, which was studied by Shi et al.~\cite{SLKS-2025}. They constructed binary self-dual codes from strongly regular graphs and extended this approach to obtain constructions of self-dual codes over finite fields from ETFs.

Shi et al.~\cite{SLKS-2025} suggested extending their results to inner products other than the standard Euclidean inner product as a direction for future research. This suggestion motivated our work. The Galois inner product is a natural generalization of the Euclidean and Hermitian inner products commonly used in coding theory and has been extensively studied in coding theory. Fan and Zhang~\cite{FZ-2017} introduced the Galois inner product and initiated the study of Galois self-dual codes. Since then, Galois constacyclic self-dual code~\cite{FL-2024,FL-2022,FL-2023}, Galois LCD codes~\cite{LFL-2018}, Galois hulls~\cite{LP-2020}, and MDS Galois self-dual codes~\cite{MC-2021} have been studied. Applications to entanglement-assisted quantum error correction have also been studied~\cite{C-2021}.

In this paper, we extend the theory of frames over finite fields developed by Greaves et al.~\cite{GIJM-2022-1,GIJM-2022-2} to the setting of the $\ell$-Galois inner product. We define the $\ell$-Galois frame operator and Gram matrix, and extend their properties from classical frame theory to the $\ell$-Galois setting. In particular, we provide necessary and sufficient conditions for the existence of a frame with a given $\ell$-Galois Gram matrix in Theorem~\ref{thm-existence}, and show that frames of the same size with the same $\ell$-Galois Gram matrix are unique up to isometry in Proposition~\ref{prop-unique}. We also extend Gerzon's bound from classical frame theory to the $\ell$-Galois setting, and show that equality in this bound yields an $\ell$-Galois equiangular tight frame in Theorem~\ref{theorem-Gerzon}. In connection with coding theory, we present constructions of $\ell$-Galois self-orthogonal and LCD codes from $\ell$-Galois tight frames. Furthermore, we generalize the result of Shi et al.~\cite{SLKS-2025} on constructing equiangular tight frames from self-dual codes to the $\ell$-Galois setting in Theorem~\ref{GETF}. Beyond their approach, we provide conditions for constructing $\ell$-Galois equiangular tight frames from $\ell$-Galois self-orthogonal or LCD codes using their automorphism groups in Theorem~\ref{thm-codes-to-frame}, and present examples satisfying these conditions. 

This paper is organized as follows. In Section 2, we provide preliminaries on the $\ell$-Galois inner product, linear codes, and their $\ell$-Galois duals. In Section 3, we present our results on $\ell$-Galois frames, including necessary and sufficient conditions for the existence of a frame with a given $\ell$-Galois Gram matrix. In Section 4, we investigate the connections between $\ell$-Galois frames and $\ell$-Galois self-orthogonal and LCD codes. We conclude our paper in Section 5.

\section{Preliminaries}
Let $\mathbb{F}_{p^e}$ be the finite field of order $p^e$, where $p$ is a prime and $e\ge 1$. For an integer $\ell$ such that $0\le \ell<e$, define an automorphism $\sigma:\mathbb{F}_{p^e}\to \mathbb{F}_{p^e}$ as
\[
\sigma(a)=a^{p^\ell}
\]
for $a\in \mathbb{F}_{p^e}$. Since $a^{p^e}=a$ for every $a\in \mathbb{F}_{p^e}$, the inverse of $\sigma$ is given by $\sigma^{-1}(a)=a^{p^{e-\ell}}$ for $a\in \mathbb{F}_{p^e}$. For a matrix $A$ over $\mathbb{F}_{p^e}$, we define
\[
A^*=\sigma(A)^T\qquad\mbox{and}\qquad A^\sharp=\sigma^{-1}(A)^T
\]
where $\sigma(A)$ and $\sigma^{-1}(A)$ denote the matrices obtained by applying $\sigma$ and $\sigma^{-1}$ entrywise to $A$, respectively. The $\ell$-Galois sesquilinear form $b_\ell$ on $\mathbb{F}_{p^e}^n$ is defined as
\[
b_\ell(u, v)=\sigma(u)^Tv
\]
for $u, v\in \mathbb{F}_{p^e}^n$. It is $\sigma$-semilinear in its first variable, linear in its second variable, and non-degenerate on both sides. 

We call a matrix $A\in \mathrm{GL}_n(\mathbb{F}_{p^e})$ an \textit{$\ell$-Galois similitude} with \textit{multiplier} $c\in \mathbb{F}_{p^e}^\times$ if
\[
b_\ell(Au, Av)=cb_\ell(u, v)
\]
for all $u, v\in \mathbb{F}_{p^e}^n$, or equivalently, $A^*A=cI_n$. The \textit{$\ell$-Galois similitude group} is defined by
\[
\Delta_\ell(n, \mathbb{F}_{p^e})=\{A\in \mathrm{GL}_n(\mathbb{F}_{p^e})~|~A^*A=cI_n\mbox{ for some }c\in \mathbb{F}_{p^e}^\times\}.
\]
A similitude with multiplier one is called an \textit{$\ell$-Galois isometry}.

A \textit{linear $[n, k]$ code} $\mathcal{C}$ over $\mathbb{F}_{p^e}$ is a $k$-dimensional subspace of $\mathbb{F}_{p^e}^n$, and its elements are called \textit{codewords}. Unless otherwise specified, the term \textit{code} refers to a linear code throughout this paper. For an $[n, k]$ code $\mathcal{C}$ over $\mathbb{F}_{p^e}$, a generator matrix $G$ of $\mathcal{C}$ is a $k\times n$ matrix over $\mathbb{F}_{p^e}$ whose rows form a basis of $\mathcal{C}$. An \textit{information set} of a linear $[n, k]$ code $\mathcal{C}$ over $\mathbb{F}_{p^e}$ is a subset $J\subseteq\{1, \ldots, n\}$ of size $k$ such that the columns of a generator matrix of $\mathcal{C}$ indexed by $J$ are linearly independent.

Two codes $\mathcal{C}_1$ and $\mathcal{C}_2$ are called \textit{(monomially) equivalent} if there is a monomial matrix $M$, that is, a square matrix with exactly one nonzero entry in each row and each column, such that $\mathcal{C}_1M=\mathcal{C}_2$. The \textit{automorphism group} $\operatorname{Aut}(\mathcal{C})$ of a code $\mathcal{C}$ over $\mathbb{F}_{p^e}$ is the group of monomial matrices that map $\mathcal{C}$ to itself. For a linear $[n, k]$ code $\mathcal{C}$ over $\mathbb{F}_{p^e}$, permuting the coordinates corresponding to an information set $J$ of $\mathcal{C}$ to the first $k$ positions yields an equivalent code with a generator matrix of the form $G=[I_k\mid A]$, which is called a \textit{generator matrix of standard form}.

Let $0\le \ell < e$. For an $[n, k]$ code $\mathcal{C}$ over $\mathbb{F}_{p^e}$, the \textit{$\ell$-Galois dual} $\mathcal{C}^{\perp_\ell}$ is defined as
\[
\mathcal{C}^{\perp_\ell}=\{\mathbf{x}\in\mathbb{F}_{p^e}^n~|~\mathbf{x}\sigma(\mathbf{c})^T=\mathbf{0}\mbox{ for all }\mathbf{c}\in\mathcal{C}\}.
\]
When $\ell=0$, the $\ell$-Galois dual coincides with the \textit{Euclidean dual}, while when $\ell=e/2$ with $e$ even, it coincides with the \textit{Hermitian dual}. If $\mathcal{C}\subseteq\mathcal{C}^{\perp_\ell}$, then $\mathcal{C}$ is called an \textit{$\ell$-Galois self-orthogonal code}, and if $\mathcal{C}=\mathcal{C}^{\perp_\ell}$, then it is called an \textit{$\ell$-Galois self-dual code}. The \textit{$\ell$-Galois hull} of $\mathcal{C}$ is defined as
\[
\operatorname{Hull}_\ell(\mathcal{C})=\mathcal{C}\cap\mathcal{C}^{\perp_\ell}.
\]
Here $\mathcal{C}$ is $\ell$-Galois self-orthogonal if $\mathcal{C}=\operatorname{Hull}_\ell(\mathcal{C})$. On the other hand, $\mathcal{C}$ is called an \textit{$\ell$-Galois LCD code} if $\operatorname{Hull}_\ell(\mathcal{C})=\{0\}$. For an $[n, k]$ code $\mathcal{C}$ over $\mathbb{F}_{p^e}$, the dimension of $\operatorname{Hull}_\ell(\mathcal{C})$ is given in~\cite{LP-2020} as follows:
\begin{equation}\label{hulldim}
\dim \operatorname{Hull}_\ell(\mathcal{C})=k-\operatorname{rank}(GG^*).
\end{equation}
The following proposition shows how the $\ell$-Galois duals of a code $\mathcal{C}$ are related to the Euclidean dual of $\mathcal{C}$.
\begin{proposition}[\cite{LP-2020}]
    Let $\mathcal{C}$ be a code over $\mathbb{F}_{p^e}$. Then
    \[
    \mathcal{C}^{\perp_\ell}=\sigma(\mathcal{C}^{\perp_0})
    \]
    for every $0\le\ell<e$.
\end{proposition}

Two square matrices $A$ and $B$ of the same size over $\mathbb{F}_{p^e}$ are called \textit{congruent} if there exists an invertible matrix $C$ such that $A=C^*BC$.
\begin{lemma}[\cite{J-1985}]\label{lem-cong}
    Let $A$ be an invertible matrix of rank $r$ over $\mathbb{F}_{p^e}$. Then the following statements hold.
    \begin{enumerate}
        \item If the order of $\sigma$ is two, and $A=A^*$, then $A\cong I_r$.
        \item If $\sigma$ is the identity, $p$ is odd and $A$ is symmetric, then $A\cong I_r$ if and only if $\det A\in (F^\times)^2$.
        \item If $\sigma$ is the identity, $p=2$ and $A$ is symmetric, then $A\cong I_r$ if and only if $A$ is non-alternating.
    \end{enumerate}
\end{lemma}

\section{\texorpdfstring{$\ell$}{l}-Galois frames}
For the remaining sections of this paper, we write $F=\mathbb{F}_{p^e}$ for the finite field of order $p^e$, where $p$ is a prime and $e\ge 1$, and let $\sigma:F\to F$ be the automorphism defined by $\sigma(a)=a^{p^\ell}$ for $a\in F$, where $\ell$ is an integer with $0\le \ell<e$.

\begin{definition}
    A sequence $\phi_1, \ldots, \phi_m$ in $F^n$ is called a \textit{frame} for $F^n$ if it spans $F^n$. The $n\times m$ matrix $\Phi$ whose columns are $\phi_1,\ldots,\phi_m$ is called a \textit{frame matrix}. We shall also refer to a frame by its frame matrix $\Phi$, and say that the frame has size $n\times m$. The \textit{$\ell$-Galois Gram matrix} of $\Phi$ is given by
    \[
    G=(b_\ell(\phi_i, \phi_j))=\Phi^*\Phi.
    \]
\end{definition}
The map $F^m\to F^n$ induced by $\Phi$ is called the \textit{synthesis operator}, while the map $F^n\to F^m$ induced by $\Phi^*$ is called the \textit{$\ell$-Galois analysis operator}. Their composition, the map $F^n\to F^n$ induced by $\Phi\Phi^*$, is called the \textit{$\ell$-Galois frame operator}. A frame $\Phi$ is called \textit{non-degenerate} if its frame operator $\Phi\Phi^*$ is invertible.

\begin{remark}\label{rem-adjoint}
    Let $\Phi$ be a frame for $F^n$ of size $n\times m$. For vectors $u\in F^n$ and $v\in F^m$, we have
    \[
    b_\ell(\Phi^\sharp u, v)=\sigma(u)^T\sigma(\Phi^\sharp)^Tv=\sigma(u)^T\Phi v=b_\ell(u, \Phi v),
    \]
    and
    \[
    b_\ell(v, \Phi^*u)=\sigma(v)^T\Phi^*u=\sigma(v)^T\sigma(\Phi)^Tu=b_\ell(\Phi v, u).
    \]
    Therefore, $\Phi^*$ is the adjoint of $\Phi$ and $\Phi$ is adjoint of $\Phi^\sharp$.
\end{remark}

\begin{proposition}\label{prop-gram-im-ker}
    Let $\Phi$ be a frame for $F^n$. Then $\ker\Phi^*\Phi=\ker\Phi$ and $\operatorname{Im}\Phi^*\Phi=\operatorname{Im}\Phi^*$.
\end{proposition}
\begin{proof}
    Let $\Phi$ be a frame of size $n\times m$. Since $\Phi$ is a frame, $\operatorname{rank}\Phi=n$. Thus, $\Phi$ induces a surjective linear map from $F^m$ to $F^n$. This yields
    \[
    \operatorname{Im}\Phi^*\Phi=\Phi^*(F^n)=\operatorname{Im}\Phi^*.
    \]
    Note that $\operatorname{rank}\Phi^*=\operatorname{rank}\sigma(\Phi)^T=n$. Therefore, $\Phi^*$ induces an injective linear map from $F^n$ to $F^m$. This implies that
    \[
    \Phi^*\Phi x=0\qquad\text{if and only if}\qquad \Phi x=0
    \]
    for any $x\in F^m$, that is, $\ker\Phi^*\Phi=\ker\Phi$.
\end{proof}

\begin{proposition}\label{prop-unique}
    Let $\Phi$ and $\Psi$ be frames for $F^n$ of the same size. Fix $c\in F^\times$. Then $\Psi^*\Psi=c\Phi^*\Phi$ if and only if there is a unique $A\in \Delta_\ell(n, F)$ such that $A^*A=cI_n$ and $\Psi=A\Phi$.
\end{proposition}
\begin{proof}
    Suppose that $\Psi=A\Phi$ for some $A\in\Delta_\ell(n, F)$. Then
    \[
    \Psi^*\Psi=\Phi^*A^*A\Phi=c\Phi^*\Phi.
    \]

    Conversely, assume that $\Psi^*\Psi=c\Phi^*\Phi$. Take any $u\in F^n$. Since the columns of $\Phi$ span $F^n$, we write $u=\Phi x$ for some $x\in F^m$. Define a map $A:F^n\to F^n$ as
    \[
    A(\Phi x)=\Psi x
    \]
    for any $x\in F^m$. We first show that $A$ is well-defined. If $\Phi x=\Phi y$ for some $x, y\in F^m$, then
    \[
    0=c\Phi^*\Phi(x-y)=\Psi^*\Psi(x-y).
    \]
    Since $\Psi^*$ has full column rank, it is injective. Thus, $\Psi x=\Psi y$, which implies that $A$ is well-defined.

    Since the image of $A$ contains all columns of the frame $\Psi$, $A$ is surjective, and hence invertible. Then, for any $u=\Phi x$ and $v=\Phi y$, we have
    \[
    b_\ell(A(u), A(v))=b_\ell(\Psi x, \Psi y)=\sigma(x)^T\Psi^*\Psi y=c\sigma(x)^T\Phi^*\Phi y=cb_\ell(u, v).
    \]
    Therefore, the matrix representing $A$ is an $\ell$-Galois similitude with multiplier $c$.

    Finally, we prove uniqueness. Let $A_1\Phi=A_2\Phi=\Psi$ for some $A_1, A_2$ such that $A_1^*A_1=A_2^*A_2=cI_n$. Then, for any $u=\Phi x\in F^n$,
    \[
    A_1u=A_1\Phi x=A_2\Phi x=A_2u.
    \]
    Thus $A_1=A_2$.
\end{proof}

Let $\tau:F\to F$ be a map defined as $\tau=\sigma^2$. For a square matrix $A$ over $F$, define its cyclic norm as
\[
N_\tau(A)=\prod_{i=0}^{s-1}\tau^i(A),\qquad s=\operatorname{ord}(\tau)=\frac{e}{\gcd(e, 2\ell)},
\]
where $\tau^0$ is the identity map. Here, the fixed field $k$ of $\tau$ is given as
\[
k=\operatorname{Fix}(\tau)=\mathbb{F}_{p^{\gcd(e, 2\ell)}}.
\]

\begin{remark}\label{rem-autord-inv}
    Assume that $r=\operatorname{ord}(\sigma)$ is odd. Write $r=2t+1$ for some integer $t$. Then
    \[
    (\sigma^2)^{t+1}=\sigma^{r+1}=\sigma.
    \]
    Therefore, we have $k=\operatorname{Fix}(\tau)=\operatorname{Fix}(\sigma)$, that is $\sigma|_k=\operatorname{id}_k$.

    On the other hand, assume that $r=\operatorname{ord}(\sigma)$ is even. For $a\in k$, we have $\sigma^2(\sigma(a))=\sigma(a)$, which gives $\sigma(a)\in k$. Thus, $\sigma|_k$ is an automorphism of $k$ and $(\sigma|_k)^2=\operatorname{id}_k$. 
    
    Write $\ell=\gcd(e, \ell)u$ for some integer $u$. Since $r=e/\gcd(e, \ell)$, we have $\gcd(r, u)=1$. This gives
    \[
    \gcd(e, 2\ell)=\gcd(e, \ell)\gcd(r, 2u)=2\gcd(e, \ell).
    \]
    Therefore,
    \[
    k=\mathbb{F}_{p^{2\gcd(e, \ell)}}\ne \mathbb{F}_{p^{\gcd(e, \ell)}}=\operatorname{Fix}(\sigma).
    \]
    Thus, $\sigma|_k$ is a nontrivial involution.
\end{remark}

The following lemma is equivalent to Hilbert's Theorem~90 for $\mathrm{GL}_m$~\cite{H-1897}.
\begin{lemma}\label{lemma-cyno-1}
    Let $A$ be a non-singular $m\times m$ matrix over $F$ such that $N_\tau(A)=I_m$. Let
    \[
    U_A=\{u\in F^m~|~A\tau(u)=u\}.
    \]
    Then the dimension of $U_A$ over $k=\operatorname{Fix}(\tau)$ is $m$. Moreover, every $k$-basis of $U_A$ is also an $F$-basis of $F^m$.
\end{lemma}
\begin{proof}
    Let $T:F^m\to F^m$ be the map defined as $T(u)=A\tau(u)$ for every $u\in F^m$. Since $T$ is a $\tau$-semilinear map, its fixed-point space $U_A$ is a $k$-subspace of $F^m$. For $v\in F^m$, define
    \[
    S(v)=\sum_{i=0}^{s-1}T^i(v)
    \]
    where $s=\operatorname{ord}(\tau)$. Since
    \[
    T^s=A\tau(A)\cdots \tau^{s-1}(A)\tau^s=N_\tau(A)=I,
    \]
    we have $T(S(v))=S(v)$ for every $v\in F^m$, that is $S(v)\in U_A$. 
    
    Let $W=\operatorname{span}_F(U_A)$. We show that $W=F^m$. Suppose not. Since $T$ preserves $W$, it induces a $\tau$-semilinear map $\overline{T}:F^m/W\to F^m/W$. Choose a nonzero $\bar{v}\in F^m/W$. For every $a\in F$, we have
    \[
    0=\overline{S(av)}=\sum_{i=0}^{s-1}\tau^i(a)\overline{T^i(v)}.
    \]
    Since $\tau^0,\tau^1,\ldots,\tau^{s-1}$ are distinct automorphisms of $F$, they are linearly independent over $F$ as functions from $F$ to itself. Therefore, $\overline{T^i(v)}=0$ for all $0\le i\le s-1$. Since $\overline{T^0(v)}=\bar{v}\ne 0$, it is a contradiction. Thus, $W=F^m$.

    Take an $F$-basis $w_1, \ldots, w_m$ of $F^m$ from $U_A$. Let $w=\sum_i\alpha_iw_i\in U_A$. Then
    \[
    \sum_i\alpha_iw_i=w=T(w)=\sum_i\tau(\alpha_i)w_i.
    \]
    Therefore, we have $\tau(\alpha_i)=\alpha_i$, that is $\alpha_i\in k$ for each $i$. This implies that $w_1, \ldots, w_m$ is a $k$-basis of $U_A$ and $\dim_k U_A=m$. Since every other $k$-basis can be obtained from $w_1,\ldots,w_m$ by an invertible change of basis over $k$, it is also an $F$-basis of $F^m$.
\end{proof}

\begin{lemma}\label{lemma-cyno-2}
        Let $H\in \mathrm{GL}_m(F)$. For $A=H^{-1}H^*$, assume $N_\tau(A)=I_m$. Let $R$ be a matrix whose columns form a $k$-basis of
        \[
        U_A=\{u\in F^m~|~A\tau(u)=u\}
        \]
        where $k=\operatorname{Fix}(\tau)$. Then, $D=R^*HR\in \mathrm{GL}_m(k)$ and $D^*=D$.
\end{lemma}
\begin{proof}
    Since both $H$ and $R$ are non-singular over $F$, $D\in \mathrm{GL}_m(F)$. Note that $H^*=HA$ and $A\tau(R)=R$ by Lemma~\ref{lemma-cyno-1}. Then we have
    \[
    D^*=R^*H^*\tau(R)=R^*HA\tau(R)=R^*HR=D.
    \]
    Since $D^*=D$, all entries of $D$ lie in $k$. Therefore, $D\in \mathrm{GL}_m(k)$.
\end{proof}

\begin{lemma}\label{lem-inv-exi}
Let $H\in\mathrm{GL}_m(F)$, $A=H^{-1}H^*$, and let
$k=\operatorname{Fix}(\tau)$. Whenever $N_\tau(A)=I_m$, choose
$R\in\mathrm{GL}_m(F)$ whose columns form a $k$-basis of
\[
    U_A=\{u\in F^m\mid A\tau(u)=u\}.
\]
Then $H=P^*P$ for some $P\in\mathrm{GL}_m(F)$ if and only if
the following conditions hold:
\begin{enumerate}
    \item[(i)] $N_\tau(A)=I_m$.
    \item[(ii)] If $\operatorname{ord}(\sigma)$ is odd and $p\ne 2$,
    then $\det H\in(F^\times)^2$.
    \item[(iii)] If $\operatorname{ord}(\sigma)$ is odd and $p=2$,
    then $R^*HR$ is non-alternating.
\end{enumerate}
In particular, if $\operatorname{ord}(\sigma)$ is even,
condition~\textup{(i)} alone is sufficient.
\end{lemma}
\begin{proof}
    Suppose that $H=P^*P$. Since $A=H^{-1}H^*=P^{-1}\tau(P)$, we have
    \[
    N_\tau(A)=P^{-1}\tau(P)\tau(P)^{-1}\cdots\tau^{s-1}(P)\tau^{s-1}(P)^{-1}\tau^s(P)=I_m.
    \]
    Assume that $p\ne 2$. If $t=\det P$, then $\det H=\sigma(t)t=t^{p^\ell+1}$ is a square in $F$. Now, we only need to consider the case in $(iii)$. Since $A\tau(R)=R$ by Lemma~\ref{lemma-cyno-1} and $A=P^{-1}\tau(P)$, we have
    \[
    \tau(PR)=PA\tau(R)=PR.
    \]
    Since $P, R\in \mathrm{GL}_m(F)$ and $\tau(PR)=PR$, it follows that $PR\in \mathrm{GL}_m(k)$. Since $\operatorname{ord}(\sigma)$ is odd, we have $\sigma|_k=\operatorname{id}_k$. It follows that
    \[
    R^*HR=(PR)^*(PR)=(PR)^T(PR).
    \]
    Therefore, $R^*HR$ is congruent to $I_m$ over $k$, and by Lemma~\ref{lem-cong}, it is non-alternating.

    We now prove the converse. Suppose that the stated conditions hold. By Lemma~\ref{lemma-cyno-1} and Lemma~\ref{lemma-cyno-2}, $R\in \mathrm{GL}_m(F)$ and $D=R^*HR\in \mathrm{GL}_m(k)$ with $D^*=D$.

    If $\operatorname{ord}(\sigma)$ is even, then the order of $\sigma|_k$ is two by Remark~\ref{rem-autord-inv}. Then, by Lemma~\ref{lem-cong}, $D\cong I_m$ over $k$.

    Suppose that $\operatorname{ord}(\sigma)$ is odd and $p\ne 2$. Since $\sigma|_k=\operatorname{id}_k$ by Remark~\ref{rem-autord-inv}, $D=D^T$. Let $t=\det R$. This gives
    \[
    \det D=\sigma(t)t\det H=t^{p^\ell+1}\det H.
    \]
    Therefore $\det D\in (F^\times)^2$. Since $\det D\in k^\times$ and $[F:k]$ is odd, we have $\det D\in (k^\times)^2$. Then, Lemma~\ref{lem-cong} gives $D\cong I_m$ over $k$.

    Finally, assume that $\operatorname{ord}(\sigma)$ is odd and $p=2$. Since $D$ is symmetric and non-alternating, by Lemma~\ref{lem-cong}, we have $D\cong I_m$ over $k$.

    In all three cases, we have $D\cong I_m$ over $k$, that is, $D=C^*C$ for some $C\in \mathrm{GL}_m(k)$. Set $P=CR^{-1}$. Then
    \[
    P^*P=(R^*)^{-1}DR^{-1}=H.
    \]
    This completes the proof.
\end{proof}
In Lemma~\ref{lem-inv-exi}, under condition~\textup{(i)}, the criterion in~\textup{(iii)} is
independent of the choice of the $k$-basis of $U_A$. Indeed, any
other such basis gives $R'=RM$ for some $M\in\mathrm{GL}_m(k)$.
When $\operatorname{ord}(\sigma)$ is odd,
\[
    (R')^*HR'=M^T(R^*HR)M.
\]
Thus $R^*HR$ is non-alternating if and only if $(R')^*HR'$ is
non-alternating.

Let $G$ be an $m\times m$ matrix over $F$ with $\operatorname{rank}G=n$. For $J\subseteq\{1, \ldots, m\}$, we write $G_{:,J}$ for the submatrix of $G$ consisting of the columns indexed by $J$. Similarly, we write $G_{J,:}$ for the submatrix of $G$ consisting of the rows indexed by $J$. If the columns indexed by $J$ form a basis of $\operatorname{Im} G$, then we call $G_{J, J}$ a \textit{basic submatrix} of $G$.

\begin{lemma}\label{lem-bas-subm}
    Let $G$ be an $m\times m$ matrix over $F$ with $\operatorname{rank}G=n$. Let $H$ be its basic submatrix, where $J\subseteq\{1, \ldots, m\}$ is the corresponding index set. If $\ker G^T=\sigma(\ker G)$, then $H$ is non-singular over $F$, and for $B=H^{-1}G_{J,:}$, the following statements hold:
    \begin{enumerate}
        \item $B_{:, J}=I_n$, 
        \item $\operatorname{rank} B=n$, and
        \item $G=B^*HB$.
    \end{enumerate}
\end{lemma}
\begin{proof}
    Let $E$ be the $m\times n$ matrix whose columns are the standard basis vectors of $F^m$ indexed by $J$, and let $K$ be an $m\times (m-n)$ matrix whose columns form a basis of $\ker G$. Suppose that $Eu+Kv=0$ for some vectors $u$ and $v$. This implies that
    \[
    0=G(Eu+Kv)=GEu,
    \]
    and since the columns of $GE=G_{:,J}$ are linearly independent, we have $u=0$, and then $v=0$. It follows that the matrix $Q=[E~|~K]$ is non-singular. Since $\ker G^T=\sigma(\ker G)$, $G^T\sigma(K)=\mathcal{O}$, that is, $K^*G=\mathcal{O}$. Then we have
    \[
    Q^*GQ=\begin{pmatrix}
        E^*GE & E^*GK\\
        K^*GE & K^*GK
    \end{pmatrix}=\begin{pmatrix}
        H & \mathcal{O}\\
        \mathcal{O} & \mathcal{O}
    \end{pmatrix}.
    \]
    Since $\operatorname{rank}Q^*GQ=\operatorname{rank}G=n$, $H$ is non-singular.

    Note that
    \[
    B_{:, J}=H^{-1}(G_{J, :})_{:, J}=H^{-1}G_{J, J}=H^{-1}H=I_n.
    \]
    Set $L=[I_n~|~\mathcal{O}]Q^{-1}$. Then we have
    \[
    G=(Q^{-1})^*\begin{pmatrix}
        H & \mathcal{O}\\
        \mathcal{O} & \mathcal{O}
    \end{pmatrix}Q^{-1}=L^*HL,\qquad\text{and}\qquad LE=LQ\begin{pmatrix}
        I_n\\\mathcal{O}
    \end{pmatrix}=I_n.
    \]
    Note that $E^*=E^T$ since $E$ has entries $0$ and $1$. Therefore, we have
    \[
    G_{J, :}=E^TG=E^*L^*HL=(LE)^*HL=HL.
    \]
    Thus, $B=L$. This completes the proof.
\end{proof}

By Lemma~\ref{lem-bas-subm}, conditions \textup{(i)} and \textup{(ii)} of the following
theorem ensure that every basic submatrix of $G$ is non-singular.
\begin{theorem}\label{thm-existence}
    Let $G$ be an $m\times m$ matrix over $F$, and $H$ be its basic submatrix. Then $G$ is an $\ell$-Galois Gram matrix of a frame for $F^n$ if and only if
    \begin{enumerate}
        \item[(i)] $\operatorname{rank}G=n$,
        \item[(ii)] $\ker G^T=\sigma(\ker G)$,
        \item[(iii)] $N_\tau(H^{-1}H^*)=I_n$, 
        \item[(iv)] if $\operatorname{ord}(\sigma)$ is odd, $\det H\in (F^\times)^2$ when $p\ne 2$, and $R^*HR$ is non-alternating when $p=2$, where $R$ is a matrix whose columns form a $k$-basis of
        \[
        U=\{u\in F^n~|~H^{-1}H^*\tau(u)=u\}
        \]
        where $k=\operatorname{Fix}(\tau)$.
    \end{enumerate}
\end{theorem}
\begin{proof}
    Suppose that $G=\Phi^*\Phi$ for a frame $\Phi$ for $F^n$. Then, by Proposition~\ref{prop-gram-im-ker}, we have
    \[
    \operatorname{rank}G=\operatorname{rank}\Phi=n\qquad\text{and}\qquad\ker G=\ker\Phi.
    \]
    Since $G^T=\Phi^T\sigma(\Phi)$ and $\Phi^T$ induces an injective linear map on $F^n$, we have
    \[
    \ker G^T=\ker \sigma(\Phi)=\sigma(\ker\Phi)=\sigma(\ker G).
    \]
    
    Let $g_1, \ldots, g_m$ and $\phi_1, \ldots, \phi_m$ be the columns of $G$ and $\Phi$, respectively. Then $g_i=\Phi^*\phi_i$. Since the map induced by $\Phi^*$ is injective, the columns of $G$ have the same linear dependence relations as the columns of $\Phi$. Let $J$ be the index set such that $H=G_{J, J}$. Since $J$ indexes a column basis of $G$, this implies that $P=\Phi_{:, J}$ is non-singular and $H=G_{J, J}=P^*P$. Then, by Lemma~\ref{lem-inv-exi}, conditions $(iii)$ and $(iv)$ follow.

    Next, we prove the converse. Let $J$ be the index set such that $H=G_{J,J}$. By Lemma~\ref{lem-bas-subm}, we have $G=B^*HB$ where $B=H^{-1}G_{J,:}$. Then, by Lemma~\ref{lem-inv-exi}, $H=Q^*Q$ for some non-singular matrix $Q$ over $F$. Define $\Phi=QB$. Then we have
    \[
    \Phi^*\Phi=B^*Q^*QB=G
    \]
    and $\operatorname{rank}\Phi=\operatorname{rank}B=n$. This completes the proof.
\end{proof}

For a given matrix $G$ over $F$, Theorem~\ref{thm-existence} provides necessary and sufficient conditions for the existence of a frame for $F^n$ whose $\ell$-Galois Gram matrix is $G$. Moreover, by Proposition~\ref{prop-unique}, such frames are unique up to isometry.

\begin{definition}
    Let $\Phi$ be a frame for $F^n$.
    \begin{enumerate}
        \item If $\Phi\Phi^*=cI_n$ for some $c\in F$, then $\Phi$ is called an \textit{$\ell$-Galois $c$-tight frame}. In this case, $c$ is called a \textit{frame constant}.
        \item If $\Phi$ is a $c$-tight frame and $(\Phi^*\Phi)_{ii}=a$ for some $a\in F$ and for every $i$, then $\Phi$ is called an \textit{$\ell$-Galois $(a, c)$-equal norm tight frame}.
        \item Let $a, b\in F$. A sequence of nonzero vectors $\phi_1, \ldots, \phi_m\in F^n$ is called an \textit{$\ell$-Galois $(a, b)$-equiangular system} if $b_\ell(\phi_i, \phi_i)=a$ for all $1\le i\le m$ and $b_\ell(\phi_i, \phi_j)b_\ell(\phi_j, \phi_i)=b$ for all $1\le i\ne j\le m$.
        \item If $\Phi$ is a $c$-tight frame, and its columns form an $\ell$-Galois $(a, b)$-equiangular system, then $\Phi$ is called an \textit{$\ell$-Galois $(a, b, c)$-equiangular tight frame}, or briefly, an $\ell$-Galois $(a, b, c)$-ETF.
    \end{enumerate}
\end{definition}

\begin{remark}
    Let $\Phi$ be an $\ell$-Galois $c$-tight frame for $F^n$. Since $\Phi\Phi^*=cI_n$, we have
    \[
    \Phi\Phi^\sharp=\Phi\sigma^{-1}(\Phi)^T=\sigma^{-1}(\Phi\sigma(\Phi)^T)^T=\sigma^{-1}(\Phi\Phi^*)^T=\sigma^{-1}(cI)^T=\sigma^{-1}(c)I.
    \]
    Let $\ell'$ be the unique integer with $0\le \ell'<e$ and $\ell'\equiv -\ell\pmod e$. Then $\Phi$ is an $\ell'$-Galois $\sigma^{-1}(c)$-tight frame for $F^n$. In particular, if $c\in \operatorname{Fix}(\sigma)$, then it is an $\ell'$-Galois $c$-tight frame.
\end{remark}

\begin{proposition}\label{prop-tightframe}
Let $\Phi$ be a frame for $F^n$. Then the following are equivalent:
\begin{enumerate}
    \item[(i)] $\Phi$ is a $c$-tight frame,
    \item[(ii)] $(\Phi^*\Phi)^2=c(\Phi^*\Phi)$,
    \item[(iii)] $b_\ell(\Phi^\sharp u, \Phi^*v)=c\cdot b_\ell(u, v)$ for every $u, v\in F^n$.
\end{enumerate}
\end{proposition}
\begin{proof}
    We first show that $(i)$ and $(ii)$ are equivalent. If $\Phi\Phi^*=cI_n$ for some $c\in F$, then we have
    \[
    (\Phi^*\Phi)^2=\Phi^*(\Phi\Phi^*)\Phi=c\Phi^*\Phi.
    \]
    Conversely, $(\Phi^*\Phi)^2=c(\Phi^*\Phi)$ means
    \[
    \mathcal{O}=(\Phi^*\Phi)^2-c(\Phi^*\Phi)=\Phi^*(\Phi\Phi^*-cI_n)\Phi.
    \]
    Since $\Phi$ is a frame, $\Phi$ induces a surjective map and since $\operatorname{rank}\Phi^*=n$, $\Phi^*$ induces an injective map. Therefore, $\Phi\Phi^*=cI_n$.

    Next, we show that $(i)$ and $(iii)$ are equivalent. By Remark~\ref{rem-adjoint}, we have
    \[
    b_\ell(\Phi^\sharp u, \Phi^*v)=b_\ell(u, \Phi\Phi^*v)=c\cdot b_\ell(u, v).
    \]
    Conversely, $(iii)$ gives
    \[
    b_\ell(u, \Phi\Phi^* v-cv)=0
    \]
    for every $u, v$. Since $b_\ell$ is non-degenerate, it follows that $\Phi\Phi^*-cI=\mathcal{O}$. This completes the proof.
\end{proof}

\begin{proposition}
    Let $\Phi$ be an $\ell$-Galois $(a, c)$-equal norm tight frame for $F^n$ of size $n\times m$. Then $nc=ma$.
\end{proposition}
\begin{proof}
    Let $\Phi=[\phi_1\cdots\phi_m]$. Then we have $\sum_{i}\phi_i\phi_i^*=cI_n$. This gives
    \[
    ma=\operatorname{tr}\left(\sum_i\phi_i\phi_i^*\right)=\operatorname{tr}(cI_n)=cn.
    \]
\end{proof}

\begin{proposition}\label{squareframe}
    Let $\Phi$ be an $\ell$-Galois $(a, b, c)$-ETF for $F^n$. If $\Phi$ is square, then $b=0$.
\end{proposition}
\begin{proof}
    Since $\operatorname{rank}\Phi=n$, it is invertible. Therefore, we have $\Phi^*\Phi=\Phi\Phi^*=cI$. This gives $b=0$.
\end{proof}

\begin{proposition}\label{prop-ETF-para}
    Let $\Phi$ be an $\ell$-Galois $(a, b, c)$-ETF for $F^n$. Then $ac=a^2+(m-1)b$.
\end{proposition}
\begin{proof}
    Let $\Phi=[\phi_1\cdots\phi_m]$, and let $A_i=\phi_i\phi_i^*$. Then we have $\sum_{i}A_i=cI_n$. Note that
    \[
    \operatorname{tr}(A_iA_j)=\operatorname{tr}(\phi_i\phi_i^*\phi_j\phi_j^*)=(\phi_i^*\phi_j)\operatorname{tr}(\phi_i\phi_j^*)=(\phi_i^*\phi_j)(\phi_j^*\phi_i),
    \]
    which is equal to $a^2$ if $i=j$ and $b$ if $i\ne j$. For each $1\le j\le m$, multiplying on the left by $A_j$ and taking trace on both sides of $\sum_{i}A_i=cI_n$ gives
    \[
    ca=\operatorname{tr}(A_j(cI_n))=\sum_{i=1}^m\operatorname{tr}(A_jA_i)=(a^2-b)+bm=a^2+(m-1)b.
    \]
\end{proof}

\begin{remark}
Let $\Phi$ be an $\ell$-Galois $(a, 0, c)$-ETF for $F^n$. By Proposition~\ref{prop-ETF-para}, $ac=a^2$. Therefore, $a=c$ or $a=0$.
\end{remark}

\begin{lemma}\label{lem-eqsystem}
    Let $\phi_1, \ldots, \phi_m\in F^n$ be an $(a, b)$-equiangular system for some $a, b\in F$ such that $a^2\ne b$. Let $A_i=\phi_i\phi_i^*$ for each $i$, and let $L:F^m\to M_{n\times n}(F)$ be the map defined by
    \[
    L(x)=\sum_{i=1}^mx_iA_i
    \]
    for $x=(x_1, \ldots, x_m)\in F^m$. Then we have the following:
    \begin{enumerate}
        \item[(i)] $\operatorname{tr}(A_i)=a$ for each $i$,
        \item[(ii)] for each $i$ and $j$, $\operatorname{tr}(A_iA_j)=a^2$ if $i=j$ and $\operatorname{tr}(A_iA_j)=b$ if $i\ne j$,
        \item[(iii)] $\ker L\subseteq \langle \mathbf{1}_m\rangle$,
        \item[(iv)] $\operatorname{rank}L\ge m-1$,
        \item[(v)] if $a\ne 0$, then $L$ is injective.
    \end{enumerate}
\end{lemma}
\begin{proof}
    Note that
    \[
    \operatorname{tr}(A_i)=\operatorname{tr}(\phi_i\phi_i^*)=\operatorname{tr}(\phi_i^*\phi_i)=a,
    \]
    and
    \[
    \operatorname{tr}(A_iA_j)=\operatorname{tr}(\phi_i\phi_i^*\phi_j\phi_j^*)=(\phi_i^*\phi_j)\operatorname{tr}(\phi_i\phi_j^*)=(\phi_i^*\phi_j)(\phi_j^*\phi_i).
    \]
    This proves $(i)$ and $(ii)$.

    Next, we prove $(iii)$ and $(iv)$. Let
    \[
    M=(\operatorname{tr}(A_iA_j))_{i, j=1}^m=(a^2-b)I_m+bJ_m
    \]
    where $J_m$ is the $m\times m$ all-one matrix. Suppose that $L(x)=\mathcal{O}$ for some $x\in F^m$. Then we have
    \[
    0=\operatorname{tr}(A_iL(x))=\operatorname{tr}\left(A_i\sum_{j=1}^mx_jA_j\right)=\sum_{j=1}^mx_j\operatorname{tr}(A_iA_j)=(a^2-b)x_i+b\sum_{j=1}^mx_j
    \]
    for every $1\le i\le m$. This gives
    \[
    (a^2-b)x_{i_1}=(a^2-b)x_{i_2}
    \]
    for every $1\le i_1, i_2\le m$, and since $a^2-b\ne 0$, it follows that $x=\alpha\mathbf{1}_m$ for some $\alpha\in F$. Since $\ker L\subseteq \langle \mathbf{1}_m\rangle$, we have
    \[
    \operatorname{rank}L=m-\dim\ker L\ge m-1.
    \]

    Finally, we show $(v)$. Suppose that $a\ne 0$ and $\ker L\ne \{0\}$. Then, by $(iii)$, $\mathbf{1}_m\in \ker L$. Thus, $\sum_{i=1}^mA_i=\mathcal{O}$, and
    \[
    ma=\operatorname{tr}\left(\sum_{i=1}^mA_i\right)=0,
    \]
    that is, $m=0$. Note that
    \[
    0=\operatorname{tr}\left(A_j\sum_{i=1}^mA_i\right)=\sum_{i=1}^m\operatorname{tr}(A_jA_i)=a^2+(m-1)b
    \]
    for every $j$. Since $m=0$, we have $a^2-b=0$, which is a contradiction. Therefore, $L$ is injective.
\end{proof}

The following theorem extends Gerzon's bound over finite fields given in~\cite[Theorem 4.2]{GIJM-2022-1} to the $\ell$-Galois setting.

\begin{theorem}\label{theorem-Gerzon}
    Let $\phi_1, \ldots, \phi_m\in F^n$ be an $\ell$-Galois $(a, b)$-equiangular system for some $a, b\in F$ such that $a^2\ne b$. Then
    \begin{equation}\label{thm-eq-1}
    m\le d_\ell(n)=\begin{cases}
        n(n+1)/2, & \mbox{if }\ell=0,\\
        n^2, & \mbox{if }0<\ell<e.
    \end{cases}
    \end{equation}
    If equality holds, then $\Phi=[\phi_1\cdots \phi_m]$ is an $\ell$-Galois $(a, b, c)$-ETF for some $c\in F$ with $c=0$ if $a=0$ and $c\ne 0$ if $a\ne 0$.
\end{theorem}
\begin{proof}
    Set $A_i=\phi_i\phi_i^*$ for each $i$ and define $L:F^m\to M_{n\times n}(F)$ as
    \[
    L(x)=\sum_{i=1}^mx_iA_i.
    \]
    Note that if $\ell=0$, then
    \[
    \operatorname{rank}L\le \dim\operatorname{Sym}_n(F)=n(n+1)/2,
    \]
    where $\operatorname{Sym}_n(F)$ is the vector space of all symmetric $n\times n$ matrices over $F$. 
    
    Assume that $a\ne 0$. Then, by Lemma~\ref{lem-eqsystem}, $L$ is injective. Therefore,
    \[
    m=\operatorname{rank}L\le d_\ell(n).
    \]
    On the other hand, if $a=0$, then $\operatorname{tr}(A_i)=0$ for every $i$. Therefore, the image of $L$ is contained in the subspace of matrices with trace zero, whose dimensions are $n(n+1)/2-1$ if $\ell=0$ and $n^2-1$ if $\ell\ne 0$. Since $m-1\le \operatorname{rank}L$ by Lemma~\ref{lem-eqsystem}, this gives $m-1\le d_\ell(n)-1$.

    Next, we show that if equality holds in~(\ref{thm-eq-1}), then $\Phi$ is an $\ell$-Galois $(a, b, c)$-ETF. Assume that $a\ne 0$. Since $m=d_\ell(n)$, we have
    \[
    \operatorname{Im}L=\begin{cases}
        \operatorname{Sym_n(F)}, & \mbox{if }\ell=0,\\
        M_{n\times n}(F), & \mbox{if }0<\ell<e.
    \end{cases}
    \]
    Since $L$ is injective, $A_1, \ldots, A_m$ form a basis of $\operatorname{Im}L$. So, there are $t_1, \ldots ,t_m\in F$ such that $I_n=\sum_{i=1}^mt_iA_i$. For each $1\le j\le m$, multiplying by $A_j$ and taking trace of both sides gives
    \[
    a=\operatorname{tr}(A_jI_n)=\sum_{i=1}^mt_i\operatorname{tr}(A_jA_i)=(a^2-b)t_j+b\sum_{i=1}^mt_i.
    \]
    Then we have $(a^2-b)t_{i_1}=(a^2-b)t_{i_2}$ for each $1\le i_1, i_2\le m$, that is, $t_1=\cdots= t_m=t$. Since $I_n=\sum_{i=1}^mt_iA_i$, we have $t\ne 0$. Therefore,
    \[
    \Phi\Phi^*=\sum_{i=1}^mA_i=t^{-1}I_n
    \]
    with $t^{-1}\ne 0$. On the other hand, assume that $a=0$. Since $\operatorname{rank}L=m-1$, we have $\dim\ker L=1$, and $\ker L=\langle \mathbf{1}_m\rangle$. Hence,
    \[
    \Phi\Phi^*=\sum_{i=1}^mA_i=L(\mathbf{1}_m)=\mathcal{O}.
    \]

    Now, it remains to prove that $\Phi$ spans $F^n$, or equivalently $\operatorname{rank}\Phi=n$. Since $\Phi\ne\mathcal{O}$, the case $n=1$ follows immediately. Assume that $n\ge 2$. Let $\operatorname{rank}\Phi=r$. Choose an $n\times r$ matrix $B$ whose columns form a basis of $\operatorname{Im}\Phi$, and write $\Phi=B\Psi$ for some matrix $\Psi$ over $F$. Let $\psi_1, \ldots, \psi_m\in F^r$ be the columns of $\Psi$. Then $A_i=\phi_i\phi_i^*=B\psi_i\psi_i^*B^*$. This implies that
    \[
    \operatorname{rank}L\le \dim\operatorname{span}\{A_1, \ldots, A_m\}\le \dim \operatorname{span}\{(\psi_i\psi_i^*)_{i=1}^m\}\le d_\ell(r).
    \]
    Then, by Lemma~\ref{lem-eqsystem}, if $\operatorname{rank}\Phi<n$, then we have
    \begin{equation}\label{eq-bound}
    d_\ell(n)-1=m-1\le \operatorname{rank}L\le d_\ell(r)\le d_\ell(n-1).
    \end{equation}
    Note that
    \[
    d_\ell(n)-d_\ell(n-1)=\begin{cases}
        n, & \mbox{if }\ell=0,\\
        2n-1, & \mbox{if }0<\ell<e,
    \end{cases}
    \]
    which is greater than $1$. Thus $d_\ell(n-1)<d_\ell(n)-1$, which contradicts~(\ref{eq-bound}). Therefore, $\operatorname{rank}\Phi=n$.
\end{proof}

\begin{example}
Here, we present two examples satisfying the conditions of Theorem~\ref{theorem-Gerzon}.
\begin{enumerate}
\item Let $F=\mathbb{F}_{3^8}$ and $\ell=2$ Then $\sigma(x)=x^9$ for $x\in F$. Since $16$ divides $81-1=80$, $\mathbb{F}_{3^4}^\times$ contains an element of order $16$, namely $h$. Define
\[
A=\begin{pmatrix}
    1 & h & 1+h & 1-h\\
    h^4 & -h^5 & h^4-h^5 & h^4 + h^5
\end{pmatrix}.
\]
Then we have
\[
A^*A=2h\begin{pmatrix}
    0 & 1 & 1 & -1\\
    -1 & 0 & -1 & -1\\
    -1 & 1 & 0 & 1\\
    1 & 1 & -1 & 0
\end{pmatrix}.
\]
Then, it follows that the columns of $A$ form a $2$-Galois $(0, 2h^2)$-equiangular system with $0^2\ne 2h^2$. Since
\[
4=d_2(2)=2^2,
\]
this satisfies the condition of Theorem~\ref{theorem-Gerzon}, and hence, $A$ is a $2$-Galois ETF. In particular, since $AA^*=\mathcal{O}$, the parameters of $A$ are $(0, 2h^2, 0)$
\item Let $F=\mathbb{F}_{5^8}$ and $\ell=3$ Then $\sigma(x)=x^{125}$ for $x\in F$. Since $3$ divides $25-1=24$, $\mathbb{F}_{5^2}^\times$ contains an element of order $3$, namely $\omega$. Define
\[
A=\begin{pmatrix}
    0 & 0 & 0 & 1 & \omega & \omega^2 & 1 & 1 & 1\\
    1 & 1 & 1 & 0 & 0 & 0 & 1 & \omega & \omega^2\\
    1 & \omega & \omega^2 & 1 & 1 & 1 & 0 & 0 & 0
\end{pmatrix}.
\]
Then we have
\[
A^*A=\begin{pmatrix}
    2 & -\omega^2 & -\omega & 1 & 1 & 1 & 1 & \omega & \omega^2\\
    -\omega & 2 & -\omega^2 & \omega^2 & \omega^2 & \omega^2 & 1 & \omega & \omega^2\\
    -\omega^2 & -\omega & 2 & \omega & \omega & \omega & 1 & \omega & \omega^2\\
    1 & \omega & \omega^2 & 2 & -\omega^2 & -\omega & 1 & 1 & 1\\
    1 & \omega & \omega^2 & -\omega & 2 & -\omega^2 & \omega^2 & \omega^2 & \omega^2\\
    1 & \omega & \omega^2 & -\omega^2 & -\omega & 2 & \omega & \omega & \omega\\
    1 & 1 & 1 & 1 & \omega & \omega^2 & 2 & -\omega^2 & -\omega\\
    \omega^2 & \omega^2 & \omega^2 & 1 & \omega & \omega^2 & -\omega & 2 & -\omega^2 \\
    \omega & \omega & \omega & 1 & \omega & \omega^2 & -\omega^2 & -\omega & 2
\end{pmatrix}.
\]
Since $(A^*A)_{ij}(A^*A)_{ji}=\omega^3=1$ for every $i\ne j$, it follows that the columns of $A$ form a $3$-Galois $(2, 1)$-equiangular system with $2^2\ne 1$. Since
\[
9=d_3(3)=3^2,
\]
this satisfies the condition of Theorem~\ref{theorem-Gerzon}, and hence, $A$ is a $3$-Galois ETF. In particular, since $AA^*=I_3$, the parameters of $A$ are $(2, 1, 1)$.
\end{enumerate}
\end{example}

\section{\texorpdfstring{$\ell$}{l}-Galois frames and codes}
Throughout this section, we retain the notation $F=\mathbb{F}_{p^e}$ and $\sigma(a)=a^{p^\ell}$ for $a\in F$, with $0\le \ell<e$.

\begin{proposition}\label{prop-frame-to-code}
    Let $\Phi$ be an $\ell$-Galois $c$-tight frame for $F^k$. Then the matrix $G=[I_k~|~\Phi]$ generates an $\ell$-Galois self-orthogonal code if $c=-1$, and an $\ell$-Galois LCD code if $c\ne -1$.
\end{proposition}
\begin{proof}
    Let $\mathcal{C}$ be a code generated by $G$. Since $\Phi\Phi^*=cI_k$, by Equation~(\ref{hulldim}), we have
    \[
    \dim\operatorname{Hull}_\ell(\mathcal{C})=k-\operatorname{rank}(GG^*)=k-\operatorname{rank}((c+1)I_k),
    \]
    which is equal to $0$ if $c\ne -1$, and $k$ if $c=-1$.
\end{proof}

\begin{example}
    Let $F=\mathbb{F}_3[x]/(x^4+x^2-1)\cong \mathbb{F}_{3^4}$, and $\ell=1$. Then 
    \[
    \sigma(a)=a^3\qquad\mbox{for}\qquad a\in F.
    \]
    Define
    \[
    A=\begin{pmatrix}
        0 & x & 1 & 0 & -x & 1\\
        1 & 0 & x & 1 & 0 & -x\\
        x & 1 & 0 & -x & 1 & 0
    \end{pmatrix}
    \]
    Then
    \[
    AA^*=\begin{pmatrix}
        0 & x & 1 & 0 & -x & 1\\
        1 & 0 & x & 1 & 0 & -x\\
        x & 1 & 0 & -x & 1 & 0
    \end{pmatrix}\begin{pmatrix}
        0 & x^3 & 1 & 0 & -x^3 & 1\\
        1 & 0 & x^3 & 1 & 0 & -x^3\\
        x^3 & 1 & 0 & -x^3 & 1 & 0
    \end{pmatrix}^T=2(1+x^4)I_3
    \]
    and
    \begin{align*}
    A^*A&=\begin{pmatrix}
        0 & x^3 & 1 & 0 & -x^3 & 1\\
        1 & 0 & x^3 & 1 & 0 & -x^3\\
        x^3 & 1 & 0 & -x^3 & 1 & 0
    \end{pmatrix}^T\begin{pmatrix}
        0 & x & 1 & 0 & -x & 1\\
        1 & 0 & x & 1 & 0 & -x\\
        x & 1 & 0 & -x & 1 & 0
    \end{pmatrix}\\
    &=\begin{pmatrix}
        1+x^4 & x^3 & x & 1-x^4 & x^3 & -x\\
        x & 1+x^4 & x^3 & -x & 1-x^4 & x^3\\
        x^3 & x & 1+x^4 & x^3 & -x & 1-x^4\\
        1-x^4 & -x^3 & x & 1+x^4 & -x^3 & -x\\
        x & 1-x^4 & -x^3 & -x & 1+x^4 & -x^3\\
        -x^3 & x & 1-x^4 & -x^3 & -x & 1+x^4
    \end{pmatrix}
    \end{align*}
    Since $(A^*A)_{ij}(A^*A)_{ji}=x^4$ for every $i\ne j$, $A$ is a $1$-Galois $(1+x^4, x^4, 2(1+x^4))$-ETF for $F^3$. Since $2(1+x^4)\ne -1$, by Proposition~\ref{prop-frame-to-code}, the code $\mathcal{C}$ generated by $G=[I_3~|~A]$ is an $[9, 3]$ LCD code, and its minimum distance is $5$.
\end{example}

We have seen that an $\ell$-Galois frame naturally gives rise to either a self-orthogonal code or an LCD code, depending on the value of $c$. Conversely, it is natural to ask under what conditions a code yields an $\ell$-Galois frame. Shi et al.~\cite{SLKS-2025} showed under what conditions a (Euclidean) self-dual code gives rise to an ETF. The following theorem extends their result to the $\ell$-Galois setting.

\begin{theorem}\label{GETF}
Let $\mathcal{C}$ be an $\ell$-Galois self-dual code of length $2n$ with $n>1$ over $F$, and let $[I~|~A]$ be a generator matrix of $\mathcal{C}$. Let $M=rA+sI+tJ$ where $r,s,t\in F$, and $J$ be the all-one matrix. Take
\[
u_i=\sum_{k=1}^nA_{ik}.
\]
Then $G=MM^*$ is the Gram matrix of an $\ell$-Galois $(a, b, c)$-ETF of size $n\times n$ if and only if one of the following conditions holds:
\begin{enumerate}
\item[(i)] Let $r\ne 0$ and $s=t=0$. In this case, $a=c=-r\sigma(r)$, $b=0$ and $\operatorname{rank}G=n$.
\item[(ii)] Let $s\ne 0$ and $r=t=0$. In this case, $a=c=s\sigma(s)$, $b=0$ and $\operatorname{rank}G=n$.
\item[(iii)] Let $s\ne 0$,  $t\ne 0$, $r=0$ and 
\[
s\sigma(t)+t\sigma(s)+nt\sigma(t)=0.
\]
In this case, $a=c=s\sigma(s)$, $b=0$ and $\operatorname{rank}G=n$.
\item[(iv)] Let $r\ne 0$, $t\ne 0$, $s=0$, and there exists $\delta\in F$ such that
\[
r\sigma(t)u_i+t\sigma(ru_i)=\delta
\]
for $1\le i\le n$,
\[
r\sigma(t)u_i+t\sigma(ru_j)+nt\sigma(t)=0,
\]
for $1\le i\ne j\le n$, and
\[
-r\sigma(r)+nt\sigma(t)+\delta\ne 0.
\]
In this case, $a=c=-r\sigma(r)+nt\sigma(t)+\delta$, $b=0$ and $\operatorname{rank}G=n$.
\item[(v)] Let $r\ne 0$, $s\ne 0$, $t=0$, and there exists $\delta\in F$ such that
\[
r\sigma(s)A_{ii}+s\sigma(rA_{ii})=\delta
\]
for $1\le i\le n$,
\[
r\sigma(s)A_{ij}+s\sigma(rA_{ji})=0
\]
for $1\le i\ne j\le n$, and
\[
s\sigma(s)-r\sigma(r)+\delta\ne 0.
\]
In this case, $a=c= s\sigma(s)-r\sigma(r)+\delta$, $b=0$ and $\operatorname{rank}G=n$.
\item[(vi)] Let $r\ne 0$, $s\ne 0$, $t\ne 0$, and there exists $\delta\in F$ such that
\[
r\sigma(s)A_{ii}+s\sigma(rA_{ii})+r\sigma(t)u_i+t\sigma(ru_i)=\delta
\]
for $1\le i\le n$,
\[
r\sigma(s)A_{ij}+s\sigma(rA_{ji})+r\sigma(t)u_i+t\sigma(ru_j)+s\sigma(t)+t\sigma(s)+nt\sigma(t)=0
\]
for $1\le i\ne j\le n$, and
\[
s\sigma(s)-r\sigma(r)+s\sigma(t)+t\sigma(s)+nt\sigma(t)+\delta\ne 0.
\]
In this case, $a=c=s\sigma(s)-r\sigma(r)+s\sigma(t)+t\sigma(s)+nt\sigma(t)+\delta$, $b=0$ and $\operatorname{rank}G=n$.
\end{enumerate}
\end{theorem}
\begin{proof}
Since $\mathcal{C}$ is an $\ell$-Galois self-dual code, we have
\[
[I~|~A][I~|~A]^*=I+AA^*=\mathcal{O},
\]
so $AA^*=-I$. We first show that $G=MM^*$ is the Gram matrix of an $\ell$-Galois $(a, b, c)$-ETF of size $n\times n$ if and only if $G=aI$ for some $a\in F^\times$. Suppose that $G=\Phi^*\Phi$ for some $\ell$-Galois $(a, b, c)$-ETF $\Phi$. Since $\Phi$ is a frame, it is invertible. Therefore, we have
\[
G=\Phi^*\Phi=\Phi^{-1}(\Phi\Phi^*)\Phi=c\Phi^{-1}\Phi=cI.
\]
As $\Phi$ is invertible, $a=c\ne 0$ and $b=0$. Conversely, suppose that $G=MM^*=aI$ with $a\ne 0$. Then $M$ is invertible and
\[
M^*M=M^{-1}(MM^*)M=aI=G
\]
Therefore, $M$ is an $\ell$-Galois $(a, 0, a)$-ETF with Gram matrix $G$.

It remains to determine exactly when $MM^*=aI$ for some $a\in F^\times$. If $r=s=0$, then $G=nt\sigma(r)J$, whose rank is at most one. Since $n>1$, this cannot be a nonzero scalar multiple of the identity matrix. Every remaining triple $(r, s, t)$ belongs to exactly one of the six cases below.
\begin{enumerate}
\item[$(i)$] Let $r\ne 0$ and $s=t=0$. Then
\[
G=r\sigma(r)AA^*=-r\sigma(r)I.
\]
So $a=c=-r\sigma(r)$, $b=0$ and $\operatorname{rank}G=n$.
\item[$(ii)$] Let $s\ne 0$ and $r=t=0$. Then $G=s\sigma(s)I$. So $a=c=s\sigma(s)$, $b=0$ and $\operatorname{rank}G=n$.
\item[$(iii)$] Let $s\ne 0$, $t\ne 0$ and $r=0$. Then $M=sI+tJ$, and
\[
G=MM^*=s\sigma(s)I+(s\sigma(t)+t\sigma(s)+nt\sigma(t))J.
\]
This is a scalar multiple of the identity matrix if and only if its off-diagonal entries vanish, or equivalently
\[
s\sigma(t)+t\sigma(s)+nt\sigma(t)=0.
\]
When this holds, $G=s\sigma(s)I\ne \mathcal{O}$. So $a=c=s\sigma(s)$, $b=0$ and $\operatorname{rank}G=n$.
\item[$(iv)$] Let $r\ne 0$, $t\ne 0$ and $s=0$. Then $M=rA+tJ$, and
\[
G=MM^*=-r\sigma(r)I+r\sigma(t)AJ+t\sigma(r)JA^*+nt\sigma(t)J.
\]
Since $(AJ)_{ij}=u_i$ and $(JA^*)_{ij}=\sigma(u_j)$, we have
\[
G_{ii}=-r\sigma(r)+r\sigma(t)u_i+t\sigma(ru_i)+nt\sigma(t)=-r\sigma(r)+nt\sigma(t)+\delta
\]
for $1\le i\le n$, and
\[
G_{ij}=r\sigma(t)u_i+t\sigma(ru_j)+nt\sigma(t)
\]
for $1\le i\ne j\le n$. Thus, $G$ is a nonzero scalar multiple of the identity matrix if and only if every $G_{ij}$ vanish for $1\le i\ne j\le n$ and there exists $\delta\in F$ satisfying the requirements in $(iv)$. In that case, $a=c=-r\sigma(r)+nt\sigma(t)+\delta$, $b=0$ and $\operatorname{rank}G=n$.
\item[$(v)$] Let $r\ne 0$, $s\ne 0$ and $t=0$. Then $M=rA+sI$, and
\begin{align*}
G=MM^*=(s\sigma(s)-r\sigma(r))I+r\sigma(s)A+s\sigma(r)A^*.
\end{align*}
This implies
\[
G_{ii}=s\sigma(s)-r\sigma(r)+r\sigma(s)A_{ii}+s\sigma(rA_{ii})
\]
for $1\le i\le n$, and
\[
G_{ij}=r\sigma(s)A_{ij}+s\sigma(rA_{ji})
\]
for every $1\le i\ne  j\le n$. Therefore $G$ is a nonzero scalar multiple of the identity matrix if and only if every $G_{ij}$ vanish for $1\le i\ne j\le n$, and the expressions $r\sigma(s)A_{ii}+s\sigma(rA_{ii})$ have a common value $\delta$ with $s\sigma(s)-r\sigma(r)+\delta\ne 0$. In that case, $a=c=s\sigma(s)-r\sigma(r)+\delta$, $b=0$ and $\operatorname{rank}G=n$.
\item[$(vi)$] Let $r\ne 0$, $s\ne 0$ and $t\ne 0$. Then $M=rA+sI+tJ$, and
\begin{align*}
    G=MM^*=&(s\sigma(s)-r\sigma(r))I+r\sigma(s)A+s\sigma(r)A^*\\&+r\sigma(t)AJ+t\sigma(r)JA^*+(s\sigma(t)+t\sigma(s)+nt\sigma(t))J.
\end{align*}
This gives
\begin{align*}
G_{ii}=&s\sigma(s)-r\sigma(r)+r\sigma(s)A_{ii}+s\sigma(rA_{ii})\\&+r\sigma(t)u_i+t\sigma(r)\sigma(u_i)+s\sigma(t)+t\sigma(s)+nt\sigma(t)
\end{align*}
for all $1\le i\le n$, and
\[
G_{ij}=r\sigma(s)A_{ij}+s\sigma(rA_{ji})+r\sigma(t)u_i+t\sigma(ru_j)+s\sigma(t)+t\sigma(s)+nt\sigma(t)
\]
for all $1\le i\ne j\le n$. Therefore, $G$ is a nonzero scalar multiple of the identity matrix if and only if every $G_{ij}$ vanish for $1\le i\ne j\le n$ and there exists $\delta\in F$ satisfying the requirements in $(vi)$. In this case, $a=c=s\sigma(s)-r\sigma(r)+s\sigma(t)+t\sigma(s)+nt\sigma(t)+\delta$, $b=0$ and $\operatorname{rank}G=n$.
\end{enumerate}
Since we have considered all six nontrivial cases, our calculation completes the proof.
\end{proof}

\begin{example}
    Let $F=\mathbb{F}_{5^4}$ and $\ell=1$. Then
    \[
    \sigma(x)=x^5\quad\mbox{for}\quad x\in F,
    \]
    and the order of $\sigma$ is four. Note that $2$ is a quadratic residue in $\mathbb{F}_{25}$. Take $\gamma\in \mathbb{F}_{25}\subset F$ such that $\gamma^2=2$. Then
    \[
    \sigma(\gamma)=\gamma^5=4\gamma=-\gamma.
    \]
    Let $B=J_4-2I_4$ where $J_4$ is the $4\times 4$ all-one matrix, and define
    \[
    A=I_{16}+\gamma\begin{pmatrix}
        -B & B & B & B\\
        B & -B & B & B\\
        B & B & -B & B\\
        B & B & B & -B
    \end{pmatrix}.
    \]
    Let $\mathcal{C}$ be the code over $F$ generated by $G=[I_{16}~|~A]$. Since
    \[
    B^*=J_4+\sigma(-2)I_4=J_4-2I_4=B,
    \]
    and $B^2=4I$, we have
    \[
    AA^*=I_{16}-2\begin{pmatrix}
        -B & B & B & B\\
        B & -B & B & B\\
        B & B & -B & B\\
        B & B & B & -B
    \end{pmatrix}^2=I_{16}-2I_{16}=-I_{16}.
    \]
    Therefore, $GG^*=\mathcal{O}$, and $\mathcal{C}$ is a $1$-Galois self-dual code.

    Set
    \[
    M=rA+sI+tJ,\qquad G=MM^*
    \]
    for some $r, s, t\in F$. Note that
    \[
    u_i=\sum_{j=1}^{16}A_{ij}=1+4\gamma
    \]
    for every $1\le i\le 16$. Now, we provide an example for each of the six cases in Theorem~\ref{GETF}.
    \begin{enumerate}
    \item[$(i)$] Take $(r, s, t)=(1, 0, 0)$. Then $M=A$, and
    \[
    G=MM^*=4I_{16}.
    \]
    Therefore, $M$ is a $1$-Galois $(4, 0, 4)$-ETF.
    \item[$(ii)$] Take $(r, s, t)=(0, 1, 0)$. Then $M=I_{16}$ and
    \[
    G=MM^*=I_{16}.
    \]
    Thus, $M$ is a $1$-Galois $(1, 0, 1)$-ETF.
    \item[$(iii)$] Take $(r, s, t)=(0, 1, 3)$. Since
    \[
    s\sigma(t)+t\sigma(s)+nt\sigma(t)=3+3+16\cdot4=0,
    \]
    the conditions in $(iii)$ of Theorem~\ref{GETF} holds. Therefore,
    \[
    G=MM^*=(I+3J)(I+3J)^*=I,
    \]
    and $M$ is a $1$-Galois $(1, 0, 1)$-ETF.
    \item[$(iv)$] Take $(r, s, t)=(1, 0, 3+2\gamma)$. Note that
    \begin{align*}
    r\sigma(t)u_i+t\sigma(ru_i)&=\sigma(3+2\gamma)(1+4\gamma)+(3+2\gamma)\sigma(1+4\gamma)\\&=(3-2\gamma)(1+4\gamma)+(3+2\gamma)(1-4\gamma)\\&=4
    \end{align*}
    for every $i$. Set $\delta=4$. Since
    \[
    r\sigma(t)u_i+t\sigma(ru_j)+nt\sigma(t)=\delta+16(3+2\gamma)(3-2\gamma)=0
    \]
    and
    \[
    -r\sigma(r)+nt\sigma(t)+\delta=-r\sigma(r)=4,
    \]
    the conditions in $(iv)$ of Theorem~\ref{GETF} holds. Then $M$ is a $1$-Galois $(4, 0, 4)$-ETF, and $G=MM^*=4I_{16}$.
    \item[$(v)$] Take $(r, s, t)=(1, 1, 0)$. Then
    \[
    r\sigma(s)A_{ii}+s\sigma(rA_{ii})=A_{ii}+\sigma(A_{ii})=1+\gamma+\sigma(1+\gamma)=2
    \]
    for every $i$. Set $\delta=2$. Then we have
    \[
    r\sigma(s)A_{ij}+s\sigma(rA_{ji})=A_{ij}+\sigma(A_{ji})=\pm(\gamma+\sigma(\gamma))=0
    \]
    for every $i\ne j$, and
    \[
    s\sigma(s)-r\sigma(r)+\delta=2.
    \]
    Therefore, the conditions in $(v)$ of Theorem~\ref{GETF} holds, and $M$ is a $1$-Galois $(2, 0, 2)$-ETF with $G=MM^*=2I_{16}$.
    \item[$(vi)$]  Take $(r, s, t)=(1, 1, 1+2\gamma)$. Note that
    \begin{align*}
    r\sigma(s)A_{ii}+s\sigma(rA_{ii})+r\sigma(t)u_i+t\sigma(ru_i)&=2+\sigma(1+2\gamma)(1+4\gamma)+(1+2\gamma)\sigma(1+4\gamma)\\&=2+(1-2\gamma)(1+4\gamma)+(1+2\gamma)(1-4\gamma)\\&=2
    \end{align*}
    for every $i$. Set $\delta=2$. Then we have
    \[
    r\sigma(s)A_{ij}+s\sigma(rA_{ji})=0,\qquad r\sigma(t)u_i+t\sigma(ru_j)=0,
    \]
    \[
    s\sigma(t)+t\sigma(s)+nt\sigma(t)=(1-2\gamma)+(1+2\gamma)+16\cdot (1+2\gamma)(1-2\gamma)=0,
    \]
    and
    \[
    s\sigma(s)-r\sigma(r)+s\sigma(t)+t\sigma(s)+nt\sigma(t)+\delta=\delta=2.
    \]
    Since the conditions in $(vi)$ of Theorem~\ref{GETF} holds, $M$ is a $1$-Galois $(2, 0, 2)$-ETF with $G=MM^*=2I_{16}$.
    \end{enumerate}
\end{example}

Although Theorem~\ref{GETF} generalizes the result of Shi et al.~\cite{SLKS-2025}, this construction using self-dual codes is still limited to $\ell$-Galois ETFs with $b=0$. To overcome this limitation, we provide conditions under which ETFs can be constructed from codes that are not necessarily self-dual.

Define
\[
\mu_\sigma=\{u\in F^\times ~:~u\sigma(u)=1\}.
\]
A $\mu_\sigma$-monomial matrix is a square matrix over $F$ which has exactly one nonzero entry in each row and column, with every such entry in $\mu_\sigma$. For such a matrix $A$, we have $A^*A=AA^*=I$.

\begin{lemma}[\cite{EGHLSVY-2011}]\label{lem-irr-rep}
    Let $V$ be a finite dimensional irreducible representation of an algebra $A$ over an algebraically closed field $k$, and let $\phi:V\to V$ be a linear operator that commutes with the action of every element of $A$ on $V$. Then $\phi=\lambda\cdot \operatorname{Id}$ for some $\lambda\in k$.
\end{lemma}

\begin{theorem}\label{thm-codes-to-frame}
    Let $\mathcal{C}$ be an $[n, k]$ code over $F$, and let $\mathcal{I}$ be an information set of $\mathcal{C}$ and $\bar{\mathcal{I}}=\{1, \ldots, n\}\backslash \mathcal{I}$. Let $\mathcal{C}'$ be a code obtained by permuting the coordinates of $\mathcal{C}$ indexed by $\mathcal{I}$ to the first $k$ positions, with the generator matrix of standard form $G=[I_k~|~A]$. Suppose that there is a subgroup $\Gamma\le \operatorname{Aut}(\mathcal{C})$ with the following properties:
    \begin{enumerate}
        \item[(i)] Every $M\in \Gamma$ is $\mu_\sigma$-monomial and permutes the coordinates indexed by $\mathcal{I}$ and $\bar{\mathcal{I}}$ only among themselves.
        \item[(ii)] For any $i, j\in \bar{\mathcal{I}}$, there is $M\in \Gamma$ which permutes the $i$-th coordinate to the $j$-th coordinate. Also, for any distinct $i_1, i_2\in\bar{\mathcal{I}}$ and distinct $j_1, j_2\in\bar{\mathcal{I}}$, there is $M'\in \Gamma$ which permutes the coordinates indexed by $\{i_1,i_2\}$ to the coordinates indexed by $\{j_1,j_2\}$, regardless of their order.
        \item[(iii)] For $M\in\Gamma$, let $M_\mathcal{I}$ denote the submatrix of $M$ consisting of the rows and columns indexed by $\mathcal{I}$. Let $\bar{F}$ be the algebraic closure of $F$. If a subspace $U\subseteq \bar{F}^k$ satisfies $M_\mathcal{I}U\subseteq U$ for all $M\in\Gamma$, then $U=\{0\}$ or $U=\bar{F}^k$.
        \item[(iv)] The code obtained by puncturing $\mathcal{C}$ at $\mathcal{I}$ is nonzero.
    \end{enumerate}
    Then $A$ is an $\ell$-Galois $(a, b, c)$-ETF for some $a, b, c\in F$.
\end{theorem}
\begin{proof}
    Without loss of generality, we may assume that $\mathcal{I}=\{1,\ldots,k\}$, since otherwise we can permute the coordinates indexed by $\mathcal{I}$ to the first $k$ positions. By $(i)$, any $M\in\Gamma$ can be written as $M=M_\mathcal{I}\oplus M_{\bar{\mathcal{I}}}$, where $M_\mathcal{I}$ and $M_{\bar{\mathcal{I}}}$ are $k\times k$ and $(n-k)\times(n-k)$ matrices, respectively. Since $M$ is a code automorphism, we have $GM=UG$ for some invertible matrix $U$ over $F$. Since
    \[
    GM=[I_k~|~A](M_\mathcal{I}\oplus M_{\bar{\mathcal{I}}})=[M_\mathcal{I}~|~AM_{\bar{\mathcal{I}}}]\qquad\mbox{and}\qquad UG=U[I_k~|~A]=[U~|~UA],
    \]
    we have $M_\mathcal{I}=U$ and $M_\mathcal{I}A=AM_{\bar{\mathcal{I}}}$. For any $y\in \operatorname{Im} A$, write $y=Ax$ for some $x\in F^{n-k}$. Then
    \[
    M_{\mathcal{I}}(y)=M_\mathcal{I}Ax=AM_{\bar{\mathcal{I}}}(x)\in \operatorname{Im}A.
    \]
    Therefore, $M_\mathcal{I}(\operatorname{Im}A)\subseteq \operatorname{Im}A$. Since $\operatorname{Im}A$ is nonzero by $(iv)$, we have $\operatorname{Im}A=F^k$ by $(iii)$. Suppose that $A$ has a zero column. Then $(ii)$ implies that every column of $A$ is zero. Therefore, all columns of $A$ are nonzero. Thus, $A$ is a frame for $F^k$.

    Take $S=AA^*$. Note that $M_\mathcal{I}^*M_\mathcal{I}=M_\mathcal{I}M_\mathcal{I}^*=I_k$ and $M_{\bar{\mathcal{I}}}^*M_{\bar{\mathcal{I}}}=M_{\bar{\mathcal{I}}}M_{\bar{\mathcal{I}}}^*=I_{n-k}$. Then we have
    \[
    M_\mathcal{I}SM_\mathcal{I}^*=(M_\mathcal{I}A)(M_\mathcal{I}A)^*=(AM_{\bar{\mathcal{I}}})(AM_{\bar{\mathcal{I}}})^*=AA^*=S.
    \]
    Therefore, $M_\mathcal{I}S=SM_\mathcal{I}$ for every $M\in\Gamma$. Then, by $(iii)$ and Lemma~\ref{lem-irr-rep}, $S=cI_k$ for some $c\in F$.

    Set $Q=A^*A$. Then we have
    \[
    M_{\bar{\mathcal{I}}}^*QM_{\bar{\mathcal{I}}}=(AM_{\bar{\mathcal{I}}})^*(AM_{\bar{\mathcal{I}}})=(M_\mathcal{I}A)^*(M_\mathcal{I}A)=A^*A=Q.
    \]
    For $M\in\Gamma$, write $M_{\bar{\mathcal{I}}}e_i=\alpha_ie_{\pi(i)}$ for some $\alpha_i\in\mu_\sigma$ and for some $\pi(i)\in \{1, \ldots, n-k\}$, where $e_i$ is the $i$-th standard basis vector of $F^{n-k}$. Then we have
    \begin{align*}
    Q_{i,j}=e_i^*Qe_j&=e_i^*M_{\bar{\mathcal{I}}}^*QM_{\bar{\mathcal{I}}}(e_j)\\&=(M_{\bar{\mathcal{I}}}(e_i))^*Q(M_{\bar{\mathcal{I}}}(e_j))\\&=(\alpha_ie_{\pi(i)})^*Q(\alpha_je_{\pi(j)})\\&=\sigma(\alpha_i)\alpha_jQ_{\pi(i), \pi(j)}
    \end{align*}
    for each $i, j$. Since $\alpha_i\in\mu_\sigma$, this implies that
    \[
    Q_{i, i}=\sigma(\alpha_i)\alpha_iQ_{\pi(i), \pi(i)}=Q_{\pi(i), \pi(i)}
    \]
    for each $i$, and
    \[
    Q_{i, j}Q_{j, i}=Q_{\pi(i), \pi(j)}Q_{\pi(j), \pi(i)}
    \]
    for each $i\ne j$. Then $(ii)$ implies that $Q_{i, i}=a$ and $Q_{i, j}Q_{j, i}=b$ for some $a, b\in F$ for every $i\ne j$. Hence, $A$ is an $\ell$-Galois $(a, b, c)$-ETF.
 \end{proof}
In Theorem~\ref{thm-codes-to-frame}, condition~\textup{(iii)} states that the
representation $M\mapsto M_\mathcal{I}$ of $\Gamma$ on $F^k$ is absolutely
irreducible.

 \begin{corollary}
     Let $\mathcal{C}$ be an $[n, k]$ code over $F$ and let $\mathcal{I}$ be an information set of $\mathcal{C}$. Suppose that the conditions of Theorem~\ref{thm-codes-to-frame} are satisfied. Then $\mathcal{C}$ is either $\ell$-Galois self-orthogonal or $\ell$-Galois LCD. Moreover, if $\mathcal{C}$ is $\ell$-Galois self-orthogonal, it yields an $\ell$-Galois $(a,b,-1)$-ETF for some $a,b\in F$, while if $\mathcal{C}$ is $\ell$-Galois LCD, it yields an $\ell$-Galois $(a,b,c)$-ETF for some $a,b,c\in F$ with $c\ne -1$.
 \end{corollary}
\begin{proof}
    Let $\mathcal{C}'$ be a code obtained by permuting the coordinates of $\mathcal{C}$ indexed by $\mathcal{I}$ to the first $k$ positions, and let 
    $G=[I_k~|~A]$ be its generator matrix of standard form. By Theorem~\ref{thm-codes-to-frame}, $A$ is an $\ell$-Galois $(a,b, c)$-ETF for some $a, b, c\in F$. Then, by Proposition~\ref{prop-frame-to-code}, $\mathcal{C}$ is either $\ell$-Galois self-orthogonal or $\ell$-Galois LCD.

    Note that $GG^*=I_k+AA^*=(1+c)I_k$. Therefore, if $\mathcal{C}$ is $\ell$-Galois self-orthogonal, then $c=-1$, while if $\mathcal{C}$ is $\ell$-Galois LCD, then $c\ne -1$.
\end{proof}

\begin{example}\label{ex-SO}
    Let $F=\mathbb{F}_{3^4}$ and $\ell=1$. Then
    \[
    \sigma(x)=x^{3}\quad\mbox{for}\quad x\in F,
    \]
    and the order of $\sigma$ is four. Note that $2$ is a quadratic residue in $\mathbb{F}_{3^2}$. Take $\gamma\in \mathbb{F}_{3^2}\subset F$ such that $\gamma^2=2$. Then
    \[
    \sigma(\gamma)=\gamma^3=2\gamma=-\gamma.
    \]
    Let $V=\mathbb{F}_2^4$ with the standard dot product. Define
    \[
    W=((-1)^{x\cdot y})_{x\in V\backslash\{0\}, y\in V}.
    \]
    The dot product is computed over $\mathbb{F}_2$, and then $-1$ is regarded as an element of $F$. Define
    \[
    A=(1+\gamma)W,\qquad\mbox{and}\qquad G=[I_{15}~|~A].
    \]
    Let $\mathcal{C}$ be a $[31, 15]$ code generated by $G$.

    We show that $\mathcal{C}$ satisfies the conditions in Theorem~\ref{thm-codes-to-frame}. We first define two types of automorphisms of $\mathcal{C}$.
    
    For $h\in V$, define a permutation matrix $P_h$ as
    \[
    P_he_y=e_{y+h}
    \]
    where $e_y$ is the standard basis vector indexed by $y\in V$ for every $y\in V$. Also, define a diagonal matrix $D_h$ as
    \[
    D_h=\bigoplus_{0\ne x\in V}((-1)^{x\cdot h}).
    \]
    Then we have
    \[
    (AP_h)_{x, y}=A_{x, y+h}=(1+\gamma)(-1)^{x\cdot(y+h)}=(-1)^{x\cdot h}A_{x, y}.
    \]
    Therefore, we have $AP_h=D_hA$, and this gives
    \[
    G(D_h\oplus P_h)=[I~|~A](D_h\oplus P_h)=[D_h~|~AP_h]=[D_h~|~D_hA]=D_hG.
    \]
    Since $D_h$ is invertible, $(D_h\oplus P_h)\in \operatorname{Aut}(\mathcal{C})$ for every $h\in V$.

    Next, for $L\in \mathrm{GL}_4(\mathbb{F}_2)$, define permutation matrices $P_L$ and $U_L$  by
    \[
    P_Le_y=e_{Ly}\qquad\mbox{and}\qquad U_Le_x=e_{(L^{-1})^Tx}
    \]
   for every $x\in V\setminus\{0\}$ and $y\in V$. Then we have
    \[
    (AP_L)_{x, y}=A_{x, Ly}=(1+\gamma)(-1)^{x\cdot L(y)}=(1+\gamma)(-1)^{L^Tx\cdot y}=A_{L^Tx, y}=(U_LA)_{x, y}
    \]
    for every $x\in V\setminus\{0\}$ and $y\in V$. This implies that $AP_L=U_LA$, and therefore, $(U_L\oplus P_L)\in \operatorname{Aut}(\mathcal{C})$ for every $L\in \mathrm{GL}_4(\mathbb{F}_2)$. Define the subgroup
    \[
    \Gamma=\langle (D_h\oplus P_h), (U_L\oplus P_L)~|~h\in V, L\in \mathrm{GL}_4(\mathbb{F}_2)\rangle\le \operatorname{Aut}(\mathcal{C}).
    \]
    By its definition, $\Gamma$ satisfies the conditions $(i)$ and $(iv)$ in Theorem~\ref{thm-codes-to-frame}. 

    The $16$ columns of $A$ are indexed by the vectors of $V$. The action of $\Gamma$ on these indices contains all translations $y\mapsto y+h$ where $h\in V$, and all invertible linear transformations $y\mapsto Ly$ where $L\in \mathrm{GL}_4(\mathbb{F}_2)$. Therefore, there is an element of $\Gamma$ which induces a permutationa sending $y$ to any element of $V$. Let $(u_1, u_2)$ and $(v_1, v_2)$ be any ordered pairs of distinct elements of $V$. Set
    \[
    d=u_2-u_1\qquad\mbox{and}\qquad e=v_2-v_1.
    \]
    Since both are nonzero, there is $L\in \mathrm{GL}_4(\mathbb{F}_2)$ such that $Ld=e$. Define $g:V\to V$ by
    \[
    g(y)=L(y-u_1)+v_1.
    \]
    Then $g(u_1)=v_1$ and $g(u_2)=L(d)+v_1=e+v_1=v_2$. Therefore, $\Gamma$ induces an action which permutes $(u_1, u_2)$ into $(v_1, v_2)$. Thus, $\Gamma$ satisfies the condition $(ii)$ in Theorem~\ref{thm-codes-to-frame}. 

    Let $Q\subset \overline{F}^{15}$ be a nonzero subspace such that
    \[
    D_hQ\subseteq Q\qquad\mbox{and}\qquad U_LQ\subseteq Q
    \]
    for every $h\in V$ and $L\in \mathrm{GL}_{4}(\mathbb{F}_2)$. Choose a nonzero vector $\omega\in Q$ of minimum support size. Suppose that $\omega$ has two nonzero coordinates $\omega_x, \omega_z$ for distinct indices $x, z\in V\backslash\{0\}$. Since $x+z\ne 0$, there exists $h\in V$ such that $x\cdot h\ne z\cdot h$. Define
    \[
    \omega'=D_h\omega-(-1)^{x\cdot h}\omega\in Q.
    \]
    Then for any $v\in V\backslash\{0\}$, we have
    \[
    \omega'_v=(-1)^{v\cdot h}\omega_v-(-1)^{x\cdot h}\omega_v.
    \]
    This gives $\omega'_x=0$ and $\omega'_z=\pm2\omega_z\ne 0$. Since $\omega'_v=0$ for every $v$ such that $\omega_v=0$, $\omega'$ has smaller support than $\omega$, which is a contradiction. Therefore, $\omega$ has only one nonzero coordinate, and $Q$ contains a standard basis vector $e_i$ for some $i\in V\backslash\{0\}$. Since the permutation action induced of $\operatorname{\mathrm{GL}}(4,2)$ on $V\backslash\{0\}$ are transitive, $Q$ contains every standard basis vector. Hence $Q=\overline F^{15}$, and this shows that $\Gamma$ satisfies the condition $(iii)$ in Theorem~\ref{thm-codes-to-frame}. Therefore, $A$ is a $1$-Galois ETF.

    In particular, $G=[I_{15}~|~A]=[I_{15}~|~(1+\gamma)W]$ where
    \[
    \begingroup
    \setlength{\arraycolsep}{2pt}
    W
    =
    \left(
    \begin{array}{rrrrrrrrrrrrrrrr}
     1&-1& 1&-1& 1&-1& 1&-1& 1&-1& 1&-1& 1&-1& 1&-1\\
     1& 1&-1&-1& 1& 1&-1&-1& 1& 1&-1&-1& 1& 1&-1&-1\\
     1&-1&-1& 1& 1&-1&-1& 1& 1&-1&-1& 1& 1&-1&-1& 1\\
     1& 1& 1& 1&-1&-1&-1&-1& 1& 1& 1& 1&-1&-1&-1&-1\\
     1&-1& 1&-1&-1& 1&-1& 1& 1&-1& 1&-1&-1& 1&-1& 1\\
     1& 1&-1&-1&-1&-1& 1& 1& 1& 1&-1&-1&-1&-1& 1& 1\\
     1&-1&-1& 1&-1& 1& 1&-1& 1&-1&-1& 1&-1& 1& 1&-1\\
     1& 1& 1& 1& 1& 1& 1& 1&-1&-1&-1&-1&-1&-1&-1&-1\\
     1&-1& 1&-1& 1&-1& 1&-1&-1& 1&-1& 1&-1& 1&-1& 1\\
     1& 1&-1&-1& 1& 1&-1&-1&-1&-1& 1& 1&-1&-1& 1& 1\\
     1&-1&-1& 1& 1&-1&-1& 1&-1& 1& 1&-1&-1& 1& 1&-1\\
     1& 1& 1& 1&-1&-1&-1&-1&-1&-1&-1&-1& 1& 1& 1& 1\\
     1&-1& 1&-1&-1& 1&-1& 1&-1& 1&-1& 1& 1&-1& 1&-1\\
     1& 1&-1&-1&-1&-1& 1& 1&-1&-1& 1& 1& 1& 1&-1&-1\\
     1&-1&-1& 1&-1& 1& 1&-1&-1& 1& 1&-1& 1&-1&-1& 1
    \end{array}
    \right).
    \endgroup
    \]
    Here $AA^*=2I_{15}$ and $A^*A=2(16I_{16}-J_{16})$. Therefore, $\mathcal{C}$ is $1$-Galois self-orthogonal code whose parameters are $[31, 15, 8]$. The diagonal entries of $A^*A$ are $2\cdot 15=0$ and all off-diagonal entries are $-2=1$. Consequently, $A$ is a $1$-Galois $(0, 1, 2)$-ETF of size $15\times 16$.
\end{example}

\begin{example}
    Let $F=\mathbb{F}_{3^5}$ and $\ell=2$. Then
    \[
    \sigma(x)=x^{9}\quad\mbox{for}\quad x\in F.
    \]
    Let $\alpha$ be a root of
    \[
    f(x)=x^5+2x+1\in\mathbb{F}_3[x]
    \]
    which is irreducible. Let $V=\mathbb{F}_2^3$ with the standard dot product. Define
    \[
    W=((-1)^{x\cdot y})_{x\in V\backslash\{0\}, y\in V}.
    \]
    The dot product is computed over $\mathbb{F}_2$, and then $-1$ is regarded as an element of $F$. Define
    \[
    A=\alpha W,\qquad\mbox{and}\qquad G=[I_{7}~|~A].
    \]
    Let $\mathcal{C}$ be a $[15, 7]$ code generated by $G$.

    For $h\in V$, define a permutation matrix $P_h$ as
    \[
    P_he_y=e_{y+h}
    \]
    where $e_y$ is the standard basis vector indexed by $y\in V$ for every $y\in V$, and define a diagonal matrix $D_h$ as
    \[
    D_h=\bigoplus_{0\ne x\in V}((-1)^{x\cdot h}).
    \]
    Also, for $L\in \mathrm{GL}_4(\mathbb{F}_2)$, define permutation matrices $P_L$ and $U_L$  by
    \[
    P_Le_y=e_{Ly}\qquad\mbox{and}\qquad U_Le_x=e_{(L^{-1})^Tx}
    \]
    Then, as shown in~\ref{ex-SO}, the subgroup
    \[
    \Gamma=\langle (D_h\oplus P_h), (U_L\oplus P_L)~|~h\in V, L\in \mathrm{GL}_4(\mathbb{F}_2)\rangle\le \operatorname{Aut}(\mathcal{C})
    \]
    satisfies the conditions in Theorem~\ref{thm-codes-to-frame}. Therefore, $A$ is a $2$-Galois ETF.

    Here, $G=[I_{7}~|~\alpha W]$ where
    \[
    \begingroup
    \setlength{\arraycolsep}{2pt}
    W
    =
    \left(
    \begin{array}{rrrrrrrr}
     1&-1&1&-1&1&-1&1&-1\\
     1&1&-1&-1&1&1&-1&-1\\
     1&-1&-1&1&1&-1&-1&1\\
     1&1&1&1&-1&-1&-1&-1\\
     1&-1&1&-1&-1&1&-1&1\\
     1&1&-1&-1&-1&-1&1&1\\
     1&-1&-1&1&-1&1&1&-1
    \end{array}
    \right).
    \endgroup
    \]
    Note that $WW^*=WW^T=8I_7$. Therefore,
    \[
    AA^*=\alpha\sigma(\alpha)WW^T=8\alpha\sigma(\alpha)I_7.
    \]
    It follows that
    \[
    GG^*=(1+8\alpha\sigma(\alpha))I_7.
    \]
    Since $\alpha^5+2\alpha+1=0$, we have
    \[
    \alpha\sigma(\alpha)=\alpha^{10}=(\alpha-1)^2=\alpha^2+\alpha+1.
    \]
    Therefore,
    \[
    1+8\alpha\sigma(\alpha)=1+2\alpha\sigma(\alpha)=2\alpha^2+2\alpha=2\alpha(1+\alpha)\ne 0
    \]
    because $\alpha\notin\{0, -1\}$. Thus, $\mathcal{C}$ is a $2$-Galois LCD code whose parameters are $[15, 7, 6]$.

    Note also that
    \[
    A^*A=\alpha\sigma(\alpha)W^TW=\alpha\sigma(\alpha)(8I_8-J_8).
    \]
    Then, its diagonal entries are $7\alpha\sigma(\alpha)=\alpha\sigma(\alpha)$ and its off-diagonal entries are $-\alpha\sigma(\alpha)$. Consequently, $A$ is a $2$-Galois $(\alpha\sigma(\alpha), (\alpha\sigma(\alpha))^2, 2\alpha\sigma(\alpha))$-ETF of size $7\times 8$.
\end{example}

\section{Conclusion}
In this paper, we extended the theory of frames over finite fields to the setting of the $\ell$-Galois inner product. We established $\ell$-Galois analogues of concepts and properties from frame theory. In particular, we provided necessary and sufficient conditions for the existence of a frame of a given size whose $\ell$-Galois Gram matrix is a given matrix. We also showed that frames of the same size with the same $\ell$-Galois Gram matrix are unique up to isometry. We further investigated the connections between $\ell$-Galois frames and $\ell$-Galois self-orthogonal and LCD codes. We presented constructions of $\ell$-Galois self-orthogonal and LCD codes from frames and, conversely, provided conditions for constructing $\ell$-Galois ETFs from such codes using their automorphism groups.

\section*{Acknowledgement}
This research (J.-L. Kim) was supported in part by the BK21 FOUR (Fostering Outstanding Universities for Research) funded by the Ministry of Education (MOE, Korea), National Research Foundation of Korea (NRF) under Grant No. 4120240415042, National Research Foundation of Korea under Grant No. RS-2024-NR121331, and by Basic Science Research Program through the National Research Foundation of Korea (NRF) funded by the Ministry of Science and ICT under Grant No. RS-2025-24534992 and Global - Learning \& Academic research institution for Master’s·PhD students, and Postdocs (LAMP) Program of the National Research Foundation of Korea (NRF) grant funded by the Ministry of Education (No. RS-2024-00441954).

\end{document}